\documentclass[sigconf]{acmart}

\AtBeginDocument{%
    
}

\copyrightyear{2024}
\acmYear{2024}
\setcopyright{rightsretained}
\acmConference[KDD '24]{Proceedings of the 30th ACM SIGKDD Conference on Knowledge Discovery and Data Mining}{August 25--29, 2024}{Barcelona, Spain} \acmBooktitle{Proceedings of the 30th ACM SIGKDD Conference on Knowledge Discovery and Data Mining (KDD '24), August 25--29, 2024, Barcelona, Spain}\acmDOI{10.1145/3637528.3672034}
\acmISBN{979-8-4007-0490-1/24/08}

\makeatletter
\gdef\@copyrightpermission{
  \begin{minipage}{0.3\columnwidth}
   \href{https://creativecommons.org/licenses/by/4.0/}{\includegraphics[width=0.90\textwidth]{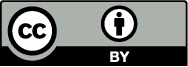}}
  \end{minipage}\hfill
  \begin{minipage}{0.7\columnwidth}
   \href{https://creativecommons.org/licenses/by/4.0/}{This work is licensed under a Creative Commons Attribution International 4.0 License.}
  \end{minipage}
  \vspace{5pt}
}
\makeatother


\usepackage{balance}
\usepackage{graphicx} 
\usepackage{bm}
\usepackage{natbib}
\usepackage{amsthm}

\theoremstyle{definition}
\newtheorem{assumption}{Assumption}

\newtheorem{proposition}{Proposition}

\usepackage{mathtools}
\mathtoolsset{showonlyrefs,showmanualtags}

\title[Learning the Covariance of Treatment Effects Across Many Weak Experiments]{Learning the Covariance of Treatment Effects\\Across Many Weak Experiments}

\author{Aur\'elien Bibaut}
\authornote{Corresponding author, contact at abibaut@netflixcontractors.com.}
\affiliation{
    \institution{Netflix}
    \city{Los Gatos}
    \state{CA}
    \country{United States}
}

\author{Winston Chou}
\affiliation{
    \institution{Netflix}
    \city{Los Gatos}
    \state{CA}
    \country{United States}
}

\author{Simon Ejdemyr}
\affiliation{
    \institution{Netflix}
    \city{Los Gatos}
    \state{CA}
    \country{United States}
}

\author{Nathan Kallus}
\affiliation{
    \institution{Netflix}
    \city{Los Gatos}
    \state{CA}
    \country{United States}
}
\affiliation{
    \institution{Cornell University}
    \city{New York}
    \state{NY}
    \country{United States}
}

\begin{document}

\begin{abstract}
When primary objectives are insensitive or delayed, experimenters may instead focus on proxy metrics derived from secondary outcomes.  For example, technology companies often infer the long-term impacts of product interventions from their effects on short-term user engagement signals. We consider the meta-analysis of many historical experiments to learn the covariance of treatment effects on these outcomes, which can support the construction of such proxies.  Even when experiments are plentiful, if treatment effects are weak, the covariance of estimated treatment effects across experiments can be highly biased. We overcome this with techniques inspired by weak instrumental variable analysis.  We show that Limited Information Maximum Likelihood (LIML) learns a parameter equivalent to fitting total least squares to a transformation of the scatterplot of treatment effects, and that Jackknife Instrumental Variables Estimation (JIVE) learns another parameter computable from the average of Jackknifed covariance matrices across experiments. We also present a total covariance estimator for the latter estimand under homoskedasticity, which is equivalent to a $k$-class estimator.  We show how these parameters can be used to construct unbiased proxy metrics under various structural models.  Lastly, we discuss the real-world application of our methods at Netflix.
\end{abstract}


\maketitle

\section{Introduction}

Projecting long-term treatment effects from short-term metrics is a ubiquitous problem in experimentation.  For example, technology companies seek to optimize insensitive primary metrics (such as habitual usage, subscriber retention, and long-term revenue), but are unable or unwilling to measure treatment effects on these metrics precisely.\footnote{For example, they may be interested in the effect of digital platform design on long-term user retention \citep{hohnhold2015focusing}, but unwilling to run a sufficiently large experiment for a long time to measure treatment effects on long-term retention.}
To address this issue, they optimize a suite of secondary metrics that are associated with the primary metrics, but more sensitive in terms of signal-to-noise and easier to measure in short-term experiments.  Under certain assumptions, the treatment effect on a ``surrogate index'' of multiple secondary metrics yields a precise estimate of the long-term treatment effect \cite{athey2016estimating,prentice1989surrogate}.

A class of statistical parameters that intuitively relates to the problem of constructing proxy metrics for a primary metric is the covariance matrix of true average treatment effects (ATEs) on primary and secondary metrics in previous experiments, and functions thereof. For example, when constructing weighted indices of secondary outcomes, it is intuitive to consider the Ordinary (OLS) and Total Least Squares (TLS) regression of true ATEs on a primary outcome on true ATEs on the secondary outcomes in the scatterplot of true ATEs over available historical experiments.  Here, by true ATE, we mean the unobserved mean on the population, in contrast to the estimated ATE actually observed on the experimental sample.

How these statistical parameters actually connect to the question of proxy metrics or surrogates is the first question we investigate in this paper. In Section \ref{sec:causal_models}, we demonstrate that statistical features of the covariance matrix of true ATEs have causal interpretations under different causal models. For instance, under these models, they can support inference on effects of novel treatments on long-term outcomes based on short-term observations.

The second question we study is the estimation of the covariance matrix of the true ATEs when the signal-to-noise ratio is small in each experiment, as is often the case in digital experimentation \citep{peysakhovich2018learning, deng2013improving, xie2016improving}. The statistical challenges are analogous to the many weak instrumental variables (IVs) setting. One question of interest in the weak instrument literature is whether allowing the number of instruments to diverge while holding their strength fixed yields consistent estimates \citep{mikusheva2022inference}. Analogously, we show that increasing the number of experiments, even if each maintains a low signal-to-noise ratio, enables consistent estimation of the covariance matrix of true treatment effects. In fact, each of the three methods we study has a weak-instrument estimator counterpart.

Our results demonstrate that we can reliably estimate the covariance matrix of true treatment effects as a parameter at the meta-analytic level, that is, from experiment-level aggregates.  In this way, we contribute to an emerging literature on meta-analytic approaches to surrogacy \cite{elliott2015surrogacy,cunningham2020interpreting, tripuraneni2023choosing}, which is particularly relevant given the large number of experiments conducted on modern online experimentation platforms.  The sheer volume of historical data from these platforms often makes unit-level analysis computationally challenging, if not prohibitive \cite{wong2019efficient}. Therefore, our methods are not only statistically robust, but also operationally feasible for large-scale experimentation. 

The paper is organized as follows:
\begin{itemize}
    \item In Section \ref{sec:data_and_statistical_setup}, we present the data collection process and the statistical parameters.
    \item In Section \ref{sec:causal_models}, we discuss the relationship of the statistical parameters to causal parameters and the construction of proxy metrics.
    \item In Section \ref{sec:estimators}, we present weak-IV-inspired estimators of the covariance matrix of true treatment effects and the OLS and TLS estimands in the scatterplot of true ATEs.
    \item In Section \ref{sec:simulations}, we conduct a simulation study to illustrate the performance of our proposed estimators and provide visual intuition on the mechanics of some of our estimators.
    \item Lastly, in Section \ref{sec:linear_models}, we describe the real-world application of our methods to derive a linear surrogate index for experimentation at Netlix.
\end{itemize}

\section{Statistical setup and notation}\label{sec:data_and_statistical_setup}

\subsection{Data}
We observe $N$ unit-level quadruples $O=(T, A, S, Y)$ where $T \in \{1,\ldots,K\}$ is the experiment index, $A \in \{0,1\}$ is the treatment arm index, $S$ is a $(G-1)$-dimensional vector of secondary metrics, and $Y$ is the primary metric of interest. We observe $N=Kn$ units divided between $K$ experiments, each of which has two arms with $n$ units.\footnote{While our results extend to the more realistic setting of multiple treatment arms of varying size per experiment, we focus on the two-arm, constant sample size case for ease of exposition.} Units in different experiments may be drawn from different superpopulations, but are assigned uniformly completely at random to cells. The observations $O_1,\dots,O_N\sim O$ are each identically distributed. They may not be independent due to the restriction of having exactly $n$ units in each experiment. Conditional on $T_1,\ldots,T_N$ and $A_1, \ldots, A_N$, any two observations are nonetheless independent.

\subsection{Treatment-effect Covariance Matrix}
For any $t=1,\ldots,K$, let $\tau_S(t) = E[S\mid A=1,T=t] - E[S\mid A=0,T=t]$ and $\tau_Y(t) = E[Y\mid A=1,T=t] - E[Y\mid A=0,T=t]$ be the true ATEs in experiment $t$ on the vector $S$ and on the scalar $Y$, respectively. 
Let 
\begin{align}
    \Lambda_K = 
    \begin{bmatrix}
    \Lambda_{YY,K} & \Lambda_{SY,K}^T \\
    \Lambda_{SY,K} & \Lambda_{SS,K}
    \end{bmatrix}
\end{align}
with
\begin{align}
    \Lambda_{YY,K} =& \frac{1}{K} \sum_{t=1}^K \tau_Y(t)^2 - \left(\frac{1}{K} \sum_{t=1}^K \tau_Y(t)\right)^2, \\
    \Lambda_{SY,K} =& \frac{1}{K} \sum_{t=1}^K \tau_S(t) \tau_Y(t)- \left( \frac{1}{K} \sum_{t=1}^K \tau_S(t) \right) \left( \frac{1}{K} \sum_{t=1}^K \tau_Y(t) \right), \\
    \Lambda_{SS,K} =& \frac{1}{K} \sum_{t=1}^K \tau_S(t) \tau_S(t)^\top - \left( \frac{1}{K} \sum_{t=1}^K \tau_S(t) \right) \left( \frac{1}{K} \sum_{t=1}^K \tau_S(t) \right)^\top.
\end{align}
That is, $\Lambda_K$ is the covariance matrix of true ATEs across the primary metric and secondary metrics over the population of experiments. We assume throughout that $\Lambda_{SS,K}$ is positive definite, meaning all its eigenvalues are strictly positive. Note that this requires at least $K\geq G-1$ tests.

\subsection{Statistical Parameters}
We consider two statistical parameters. Our first parameter is 
\begin{align}
    \theta_1(\Lambda_K) = \Lambda_{SS,K}^{-1} \Lambda_{SY,K}.
\end{align}
This is the OLS in the scatterplot of the $K$ true treatment effects.  While we discuss the exact causal interpretation of this parameter (which requires structural assumptions) in the next section, intuitively the OLS measures the statistical relationship between the surrogate metrics and the long-term metric \citep{tripuraneni2023choosing}.

Our second parameter is more complex:
\begin{align}
    \theta_{2, \Psi}(\Lambda_K) = -\gamma_{\Psi,S}(\Lambda_K) / \gamma_{\Psi,Y}(\Lambda_K), 
\end{align}
where $[\gamma_{\Psi,Y}(\Lambda_K),  \gamma_{\Psi,S}(\Lambda_K)^T]^T$ is a generalized eigenvector solving $(\Lambda_K - \kappa \Psi) \gamma = 0$ for the smallest possible $\kappa\geq0$ for which a solution exists, where $\Psi$ is a given positive definite matrix.  The parameter $\theta_{2,\Psi}(\Lambda_K)$ is the TLS on the $\Psi^{-1/2}$-transformed scatterplot of true treatment effects.  While this parameter is less intuitive, we show below that it coincides with $\theta_1(\Lambda_K)$ under certain structural assumptions.

We emphasize that these OLS and TLS parameters are defined in terms of only $K$ points in a $G$-dimensional scatterplot. The location of these points themselves (the true ATEs) are, nonetheless, unknown population quantities. Therefore, $\theta_1(\Lambda_K)$ is a population quantity, meaning it is a function of the unknown population distribution of the data $O$ (that is, of infinite draws of $O$).

\section{Causal Models and Relationship to Proxies}\label{sec:causal_models}

We now consider various causal models and examine how the statistical parameters $\theta_1(\Lambda_K)$ and $\theta_2(\Lambda_K)$ relate to causal parameters in these models.  Equivalently, this section derives the structural assumptions under which these statistical parameters can be given a causal interpretation in terms of the effect of $S$ on $Y$ and thereby used to construct proxy metrics for $Y$.

\subsection{Linear Structural Model without Direct Effects}
\label{sec:no-direct-effects}

Consider the following linear structural model with experiment-level fixed-effects and unmeasured confounding $U$ between $S$ and $Y$:
\begin{align}
    S &= \mu_S(T) + \pi_S(T) A + \gamma U + \eta \\
    Y &= \mu_Y(T) + S \beta + \delta U + \epsilon.
\end{align}
Here, $\pi_S(T)$ is the $(G - 1) \times 1$ vector of first-stage effects in experiment $T$ (that is, the ATEs on the $G - 1$ secondary metrics $S$ of the intervention trialed in experiment $T$), $\mu_S(T)$ is a $(G - 1) \times 1$ vector of per-experiment fixed effects on $S$, and $\mu_Y(T)$ is a scalar per-experiment fixed effect on $Y$.

Suppose that $A$ is randomized, and that the errors are such that $E[\eta \mid S, U, A] = 0$ and $E[\epsilon \mid S, U, A] = 0$. Then the statistical estimands $\tau_S(t)$ and $\tau_Y(t)$ identify the causal parameters $\pi_S(t)$ and $\pi_S(t) \beta$. It is straightforward to check that under the data-generating process induced by this structural model,
\begin{align}
    \beta = \theta_1(\Lambda_K) = \theta_{2,\Psi}(\Lambda_K)
\end{align}
for any positive definite $\Psi$.
The equality of the two statistical parameters can be understood as follows. Under the DGP induced by the above structural model, in which $S$ fully mediates the effect of $A$ on $Y$, $\tau_Y(t) = \tau_S(t) \beta$ for every $t \in [K]$.  This implies that $[1, -\beta]^T$ is an eigenvector associated with eigenvalue 0, which must be the smallest as $\Lambda_K$ is positive semi-definite.  Therefore, the TLS estimand is also the OLS estimand.

\paragraph{Relationship to proxies.} Now consider a new experiment $T=K+1$, in which we also randomize units equiprobably between treatment ($A=1$) and control ($A=0$). Denoting $S(t,a)$, $Y(t,a)$ the potential outcomes generated by the structural model (that is, the random variables obtained by setting values $T=t$ and $A=a$ and sampling $U$ and the error terms in the above equations), the ATE on $Y$ is related to the difference in arm-specific means on $S$ as follows: 
\begin{align}
    &E[Y(K+1,1) - Y(K+1,0)] \\
    &= E[S(K+1,1) - S(K+1,0)] \beta \\
    &= (E[S \mid T=K+1, A=1] - E[S \mid T=K+1, A=0]) \theta_i(\Lambda_K),
\end{align}
where $i =1$ or $i=(2,\Psi)$.
In words, given the data-generating distribution in $K$ historical experiments in which we observe $Y$ and $S$, and the data-generating distribution in an experiment in which we observe only the short-term outcome $S$, we can estimate the ATE on the long-term outcome. That is, $h(S) = \theta_i(\Lambda_K) S$ is an unbiased surrogate index for $Y$, meaning that ATEs on $h(S)$ equal ATEs on $Y$.

\subsection{Linear Structural Model with INSIDE-consistent Direct Effects}

\label{sec:direct-effects}

We now enrich the above linear structural model with direct effects $\pi_Y(T)$ on $Y$ that are not mediated by $S$:
\begin{align}
    S &= \mu_S(T) + \pi_S(T) A + \gamma U + \eta \\
    Y &= \mu_Y(T) + \pi_Y(T)A +  S \beta + \delta U + \epsilon.
\end{align}
It is generally impossible to disentangle the direct ($\pi_Y(T)$) and indirect/mediated ($\pi_S(T)\beta$) effects.  Still, it is possible to estimate meaningful causal parameters under the assumption of INstrument Strength Independent of Direct Effect (INSIDE) from the Mendelian randomization literature \cite{burgess2017interpreting}, which states that the first-stage effects of $A$ on $S$ are independent of the direct effects of $A$ on $Y$. In particular, INSIDE requires that the vector $\bm \pi_Y = [\pi_Y(1), \ldots, \pi_Y(K)]^\top$ of direct effects is orthogonal to the columns of the matrix $\bm{\Pi}_S = [\pi_S(1)^\top,\ldots,\pi_S(K)^T]^\top$ of first-stage effects on $S$, that is, $\bm\pi_Y^T \bm \Pi_S = 0$.

Under the data-generating distribution induced by randomized treatment assignment and the INSIDE-consistent causal model, the covariance matrix $\Lambda_K$ of true treatment effects is
\begin{align}
    \Lambda_K = 
    \begin{bmatrix}
        \beta^T \Lambda_{SS,K} \beta  & (\Lambda_{SS,K} \beta)^\top \\
        \Lambda_{SS,K} \beta & \Lambda_{SS,K}
    \end{bmatrix}
    + 
    \begin{bmatrix}
        \frac{\pi_Y^\top\pi_Y}{K} - \frac{(\pi_Y^\top \bm{1})^2}{K^2} & 0 \\
        0 & 0,
    \end{bmatrix}
\end{align}
with $\Lambda_{SS,K} = \frac{\Pi_S^\top \Pi_S}{K} - \frac{\bm{1}^\top \Pi_S \Pi_S^\top \bm 1}{K^2}$.

As can be seen from the above matrix expression, $\theta_1(\Lambda_K)=\beta$ identifies the structural parameter $\beta$. However, with direct effects, it is no longer the case in general that $[1, -\beta]^\top$ is an eigenvector of $\Lambda_K$ and therefore the TLS estimand $\theta_{2,\Psi}(\Lambda_K)$ diverges in general from $\beta$.

\paragraph{$S$-mediated ATE} In the presence of direct effects, it is no longer the case that ATEs on $Y$ are given by ATEs on $S$ times $\beta$, as $S$ no longer mediates the effect of the treatment fully. Still, $\beta$, which coincides with the OLS estimand $\theta_1(\Lambda_K)$, has a causal interpretation under the linear structural model. Let us introduce the potential outcomes $S(t,a)$ and $Y(t,a,s)$ generated by the structural model (that is, the random variables obtained by setting $T=t,A=a,S=s$ in the above equations and independently sampling $U$ and the error terms). It holds that the natural indirect effect on $Y$ through $S$ is given by 
\begin{align}
    &E[Y(K+1, a, S(K+1, 1)) - Y(K+1, a, S(K+1, 0)]\\
    &= E[ S(K+1, 1) - S(K+1, 0)] \beta \\
    &= (E[S \mid T=K+1, A=1] - E[S \mid T=K+1, A=0]) \theta_1(\Lambda_K),
\end{align}
for either $a\in\{0,1\}$.
That is, we can identify from the population distribution of historical experiments a proxy $h(S) = \theta_1(\Lambda_K)S$ such that in a new experiment, the ATE on $h(S)$ equals the part of the effect of the intervention on $Y$ that is mediated by $S$.

\subsection{Nonparametric IV Model}
\label{sec:npiv}
Relaxing the assumption of linearity, we now consider the following nonparametric IV model:
\begin{align}
    S &= \mu_S(T) + \pi_S(T) A + \eta \\
    Y &= \mu_Y(T) +  h(S) + \epsilon,
\end{align}
where $E[\eta \mid A, T] = 0$ and $E[\epsilon \mid A, T] = 0$.
We now show that, under some assumptions, $\theta_1$ identifies
a functional of $h$.
\begin{assumption}[Small effects]\label{asm:small_effects}
    $\left\lVert \pi_S(t) \right\rVert_\infty \leq \epsilon$ for every $t \in [K]$ for some $\epsilon >0$.
\end{assumption}
We will think of $\epsilon$ as a small quantity, which is often reasonable in digital experiments.
\begin{assumption}[Bounded Hessian]\label{asm:bounded_hessian}
    $h$ is twice-differentiable and $\left\lVert \nabla^2 h \right\rVert \leq M$ for some $M >0$, where the norm is the nuclear norm.
\end{assumption}
\begin{proposition}
    Suppose that Assumptions 1-2 hold. Then, 
    \begin{align}
        \theta_1(\Lambda_K) =& \left(\sum_{t=1}^K \pi_S(t) \pi_S(t)^\top \right)^{-1} \sum_{t=1}^K  \pi_S(t) \pi_S(t)^\top E[\nabla h(S(t,0))] \\
        &+ O(M \epsilon),
    \end{align}
where $S(t,0)$ is the potential outcome generated by the structural model under the control treatment in experiment $t$. 
\end{proposition}

\begin{proof}
    By definition of $\theta_1(\Lambda_K)$, the randomization of $A$, and from a second-order Taylor expansion,
    \begin{align}
        &\theta_1(\Lambda_K) \\
        =& \left(\sum_{t=1}^K \tau_S(t) \tau_S(t)^\top \right)^{-1} \sum_{t=1}^K  \pi_S(t) E[h(S(t,0) + \pi_S(t)) - h(S(t,0))] \\
        =& \left(\sum_{t=1}^K \pi_S(t) \pi_S(t)^\top \right)^{-1} \sum_{t=1}^K \pi_S(t) \pi_S(t)^\top E[\nabla h (S(t,0))] \\
        &+ \left(\sum_{t=1}^K \pi_S(t) \pi_S(t)^\top \right)^{-1} \sum_{t=1}^K \pi_S(t) \pi_S(t)^\top E[\nabla^2 h(\widetilde{S}(t,0)) \tau_S(t)],
    \end{align}
    for some $\widetilde S(t,0)$ on the segment $[S(t,0), S(t,0) + \tau_S(t)]$.
    Assumptions \ref{asm:small_effects} and \ref{asm:bounded_hessian} then yield that the second term above is $O(M \epsilon)$.
\end{proof}

In words, Proposition 1 above tells us that under the NPIV model above, the OLS statistical estimand identifies an instrument-strength-weighted average of the expected gradient of the structural function $h$.

\section{Estimating Treatment Effect Covariances}\label{sec:estimators}

In what follows, we will concatenate the variables $S$ and $Y$ in the vector $D = [Y, S]$ and also write $\tau(t)=[\tau_Y(t),\tau_S(t)]$.

\subsection{A Naive Estimator}

A naive estimator for the covariance matrix $\Lambda_K$ of true treatment effects is simply the empirical covariance matrix $\hat \Sigma_K$ of the estimated ATEs $\widehat \tau_S(t)$ and $\widehat \tau_Y(t)$:
\begin{align}
    \widehat \Sigma_K = \frac{1}{K} \sum_{t=1}^K \widehat \tau(t) \widehat \tau(t)^\top - \left(\frac{1}{K} \sum_{t=1}^K \widehat \tau(t)\right) \left(\frac{1}{K} \sum_{t=1}^K \widehat \tau(t)\right)^\top,
\end{align}
with $\widehat \tau(t) = \frac{1}{n/2} \sum_{i=1: A_i=1, T_i=t} D_i - \frac{1}{n/2} \sum_{i=1: A_i=0, T_i=t} D_i$.

What is the target estimand of $\widehat \Sigma_K$? The total variance formula yields that:
\begin{align}
    E \widehat \Sigma_K = \Lambda_K + \frac{4}{n} \bar\Omega_K, 
\end{align}
where 
\begin{align}
    \bar\Omega_K =& \frac{1}{K} \sum_{t=1}^K  \frac{\Omega_{t,1} + \Omega_{t,0}}{2}
\end{align}
and $\Omega_{t,a} = \mathrm{Cov}(D \mid T=t, A=a)$ is the within-cell covariance matrix of $D$. In other words, $4 \bar\Omega_K / n$ is the covariance of the unit-level sampling variance, or ``noise.''

In industrial (especially digital) experimentation, ATEs typically exhibit a low signal-to-noise ratio. This implies that the term $4 \bar\Omega_K / n$ is often significant relative to $\Lambda_K$ \cite{cunningham2020interpreting,peysakhovich2018learning}. As a result, under our first or second structural model, the estimate of $\beta$ based on the empirica l covariance matrix, $\widehat \beta  = \widehat{\Sigma}_{SS,K}^{-1} \widehat{\Sigma}_{SY,K}$, is biased and remains inconsistent for $\beta$ even as we increase the number of experiments (that is, as $K\to\infty$).  Since the two-stage least squares estimator under a categorical instrument equals OLS on the group means, it holds that 
$\widehat \beta$ is 2SLS with
\begin{itemize}
    \item dependent variable $(2A-1) Y$,
    \item endogenous variable $(2A-1) S$,
    \item instrument $T$.
\end{itemize}
The fact that $\widehat \beta$ is biased when noise is non-negligible reflects the well-established fact that 2SLS is inconsistent under weak instruments \cite{angrist1999jackknife}.
This suggests that insights from the weak instrumental variable literature could improve the estimation of $\Lambda_K$. In the next sections, we will consider three different methods inspired by weak IV estimators: (1) the Jackknife Instrumental Variable estimator (JIVE), (2) the Limited Information Maximum Likelihood (LIML) estimator, and (3) the general form of $k$-class estimators.

\subsection{Jackknifed Covariance Matrix of Treatment Effects}

We propose an estimator of $\Lambda_K$ inspired by the JIVE estimator \cite{angrist1999jackknife}. Consider the transformed vector $\widetilde D = 2(2A - 1) D$. Observe that $E[\widetilde D \mid T = t]  =\tau(t)= E[D \mid T=t, A=1] - E[E \mid T=t, A=0]$.

For any $t,a,i$, let 
\begin{align}
    \widehat \tau(t) = \frac{1}{n} \sum_{i:T_i  = t} \widetilde D_i \qquad \text{and} \qquad \widehat \tau_{-i}(t) = \frac{1}{n-1} \sum_{\substack{j \neq i \\ T_j = t}} \widetilde D_j 
\end{align}
be the estimated ATE on $D$ in experiment $t$, and its counterpart that leaves out observation $i$. Let
\begin{align}
    \widehat\Lambda_2 (t) = \frac{1}{n} \sum_{i: T_i = t} \widehat \tau_{-i}(t) \widetilde D_i ^\top
\end{align}
be the Jackknifed second-order moments matrix in experiment $t$, and let
\begin{align}
    \widehat \Lambda_K^{\mathrm{JK}} = \frac{1}{K} \sum_{t=1}^K \widehat\Lambda_2 (t) - \left(\frac{1}{K} \sum_{t=1}^K \widehat \tau(t)\right) \left( \frac{1}{K} \sum_{t=1}^K \widehat \tau(t) \right)^\top
\end{align}
be the Jackknifed covariance matrix. The Jackknife construction ensures that only the common source of variation between units in the same experiments is captured: the variation due to the treatment assignment, as opposed to unit-level noise. It is immediate that the Jackknifed second-order moment matrix is an unbiased estimate of the second-order moment matrix of the true treatment effects. It can also readily be checked that under mild conditions the second order term has bias $O(N^{-1})$; in other words, the bias scales with the total number of units in all experiments as opposed to the number of observations per experiment. We state this formally in the following proposition:
\begin{proposition}
    Suppose that $\max_{t,a}\left\lVert \Omega_{t,a} \right\rVert \leq M$ for some $M >0$, where the norm is any matrix norm. Then $E[\widehat \Lambda_K^{\mathrm{JK}}] = \Lambda_K + O(M/N)$.
\end{proposition}

\begin{proof}
    As the second-order moment Jackknifed matrix $\allowbreak K^{-1} \allowbreak \sum_{t=1}^K \widehat \Lambda_2(t)$ is an unbiased estimate of $K^{-1} \sum_{t=1} \tau(t) \tau(t)^\top$, the bias reduces to that of the second term, which we rewrite as $N^{-1} Z_N Z_N^\top$, where $Z_N = N^{-1/2} \sum_{i=1} \widetilde D_i - E[\widetilde D_i]$.  The (sequence of) random variables $Z_N$ can be controlled under various set of conditions. For example, suppose $0 \leq m \preccurlyeq \Omega_{t,a} \preccurlyeq M < \infty$ for every $t,a$, and suppose  that $K^{-1} \sum_{t} \Omega_{t,1} + \Omega_{t,0} \to \Omega$ in probability for some $\Omega$. Then, under a Lyapunov condition, a standard central limit theorem guarantees that $Z_N \rightsquigarrow Z$ where $Z \sim \mathcal{N}(0,\Omega)$. Then $E[(Z_N Z_N^\top)] \to E[Z Z^\top)] \preccurlyeq M I_G$, which implies the result.
\end{proof}

As a corollary, if it admits a probability limit, the plug-in estimator $\theta_i(\widehat \Lambda_K^{\mathrm{JK}})$ for either $i =1$ or $i=(2,\Psi)$ is a consistent estimator of the parameter $\theta_i(\Lambda_K)$ as the number of experiments grows, $K \to \infty$, even as the size of each experiment remains fixed. The former is equivalent to the JIVE estimator, while the latter is equivalent to the estimator proposed in \cite{hahn2004estimation}.

Note that the Jackknifed covariance matrix estimator does not require homoskedasticity of the unit-level noise --- that is, it does not require $\Omega_{a,t} = \Omega$ for all $t, a$ --- for consistency.  However, it does require us to Jackknife the unit-level data in every experiment, which can be computationally prohibitive when $n$ and $K$ are large.  In the next two sections, we therefore consider how we can leverage homoskedasticity when it is a reasonable assumption.

\subsection{Estimating Treatment Effect Covariances by Isotropizing Noise}
\label{sec:tls}

In this section and the next, we assume a common noise covariance matrix $\Omega$ across all experiments and treatment arms, that is, $\Omega_{t,a} = \Omega$ for every $t,a$. We will further assume that we know $\Omega$ to a high relative precision. These are often reasonable assumptions in digital experiments: while treatment effects are small, correlations across metrics tend to be (1) stable across experiments, and (2) non-negligible, and thus easy to estimate with a high signal-to-noise ratio by leveraging populations across multiple experiments (e.g., the entire user base). Furthermore, considering metrics that are sufficiently nonredundant ensures that $\Omega$ is well-conditioned.

Under the homoskedasticity assumption, the total covariance formula yields that
\begin{align}
    E\widehat \Sigma_K = \Lambda_K + \frac{4}{n} \Omega.
\end{align} 
Under known $\Omega$, we can multiply $\widehat \Sigma_K$ on both sides by $\Omega^{-1/2}$ to obtain a transformed $\Lambda_K$ plus isotropic noise:
\begin{align}
    E\Omega^{-1/2}\widehat \Sigma_K \Omega^{-1/2} = \Omega^{-1/2}\Lambda_K\Omega^{-1/2} + \frac{4}{n} I_G.
\end{align} 
Because adding a multiple of the identity to a matrix does not affect its eigenvectors nor the rank of their corresponding eigenvalues, the smallest eigenvector of $\Omega^{-1/2}\Lambda_K\Omega^{-1/2}$ is the smallest eigenvector of $E\Omega^{-1/2}\widehat \Sigma_K \Omega^{-1/2}$ (where by smallest eigenvector we mean an eigenvector assoicated with the smallest eigenvalue). Denote this eigenvector by $\widetilde \gamma_K$. Because $\Lambda_K$ and $\Omega$ both have non-negative eigenvalues, applying the transformation $\Omega^{-1/2}$ also does not affect the ordering of their eigenvectors, so we can recover the smallest eigenvector of $\Lambda_K$ by applying $\Omega^{-1/2}$ to $\widetilde \gamma_K$. Denote by $\widehat \gamma_K$ the result of this procedure applied to the estimate $\widehat \Sigma_K$ itself (that is, rather than its unknown expectation), and let $\widehat \theta^{LIMLK} = -\widehat \gamma_{K,S} / \widehat \gamma_{K,Y}$. Under the existence of the appropriate probability limits, $\widehat \theta^{LIMLK}$ is a consistent estimate of $\theta_{2,\Omega}(\Lambda_K)$.

As our notation suggests, $\widehat \theta^{LIMLK}$ is equal to the LIMLK (LIML with Known noise covariance matrix) estimate with dependent variable $\widetilde Y= 2(2A-1)Y$, endogenous predictors $\widetilde S = 2(2A-1)S$, and instrument $T$. (This can be observed directly from the definition of LIMLK \cite{anderson2009limited}.)

As the smallest eigenvector of the covariance matrix of observations is the statistical target of Total Least Squares (TLS), the procedure we just described, and thus LIMLK, can be implemented by: (1) transforming the observations by applying $\Omega^{-1/2}$, (2) running TLS in the transformed space, and (3) transforming the smallest eigenvector obtained from TLS by applying $\Omega^{-1/2}$.  We illustrate this procedure visually in Section~\ref{sec:simulations}.
 
A causal inference implication of this method is that, under the linear structural model presented in subsection 3.1, the structural coefficient $\beta$ equals $\theta_{2,\Omega}(\Lambda_K)$, which we can estimate consistently with $\widehat \theta^{LIMLK}$. However, under the presence of direct effects, as mentioned in Section \ref{sec:direct-effects}, $\theta_{2,\Omega}(\Lambda_K)$ no longer equals $\beta$, and therefore treatment effects on $\theta_{2,\Omega}(\Lambda_K)S$ cannot be interpreted as the part of the treatment effect of $A$ on $Y$ that is mediated by $S$.

\subsection{Estimating Treatment Effect Covariances by Subtracting Noise}

In the previous subsection, we leveraged only the direction of the known $\Omega$ but not its scale, allowing us to make a connection to LIMLK. (Note that the above procedure would give the same result if we used $\rho \Omega$ for $\rho > 0$ instead of $\Omega$.) Under known $\Omega$, there is a perhaps more straightforward estimation procedure for $\Lambda_K$: subtract $(4/n) \Omega$ from $\widehat \Sigma_K$. Formally, letting $\widehat \Lambda_K^{\mathrm{TC}} = \widehat \Sigma_K - (4/n) \Omega$ (where TC stands for Total Covariance), we have that
\begin{align}
    E \widehat \Lambda_K^{\mathrm{TC}}  = \Lambda_K,
\end{align}
and therefore, under the existence of the appropriate probability limits, the plug-in estimator $\theta_i(\widehat \Lambda_K^{\mathrm{TC}})$, for either $i =1$ or $i=(2,\Psi)$, provides a consistent estimator for $\theta_i(\Lambda_K)$.  In particular, when either $S$ fully mediates the effect of $A$ on $Y$ (no direct effects) or when direct effects follow the INSIDE assumption, we can consistently estimate the structural parameter $\beta$ with $\theta_1(\widehat \Lambda_K^{\mathrm{TC}})$.

\paragraph{Connection with IV-estimators.} Defining the matrix of centered observations  $\widetilde D^0 = [\widetilde Y^0, \widetilde S^0]$ as $\widetilde D^0 = D (I -N^{-1} \bm{1} \bm{1}^\top) \widetilde D$, and using the empirical within-experiment covariance for $\Omega$, we can readily check that 
\begin{align}
    \theta_1(\widehat \Lambda_K^{\mathrm{TC}}) = \frac{ (\widetilde S^0)^\top (I -  (1 + 4/n) M_T) \widetilde Y^0}{(\widetilde S^0)^\top (I -  (1 + 4/n) M_T) \widetilde S^0}, \label{eq:k-class}
\end{align}
with $M_T = I - P_T$, $P_T = \widetilde{T} (\widetilde{T}^\top \widetilde{T})^{-1} \widetilde{T}^\top$, and  $\widetilde T$ is the $N \times K$ matrix of one-hot encodings of experiment membership $T$ (that is $\widetilde T_{i,t} = \bm{1}\{ T_i = t \}$). One might recognize from \eqref{eq:k-class} that $\theta_1(\widehat \Lambda_K^{\mathrm{TC}})$ is a so-called $k$-class IV estimator, with $k = 1 + 4 / n$.

\section{Simulation Study}\label{sec:simulations}

To provide insight into our statistical setup and the performance of our estimators, we conduct a simulation study.  The parameters of our simulations are chosen to reflect aspects of real-world data.  The unit-level noise, denoted by $\Omega$ in our statistical setup, is typically large relative to the variance-covariance matrix of treatment effects, denoted by $\Lambda_K$, and is also potentially anti-correlated.  For example, clicks may be positively correlated with conversions outside of any experiment, but treatments that increase ``click-baitiness'' can reduce conversions.

Following this example, the top panel of Figure~\ref{fig:setup} depicts a situation where treatment effects on a proxy metric $S$ are negatively correlated with treatment effects on a primary objective $Y$.  Throughout the figure, the white arrow points in the direction of the covariance in true treatment effects.  In the middle panel, we plot the unit-level noise, which is positively correlated between $S$ and $Y$, represented by the black line.  Lastly, in the bottom panel, we plot the sampling distribution of treatment effect estimates with a fixed sample size.

\begin{figure}[h]
    \includegraphics[width=2in]{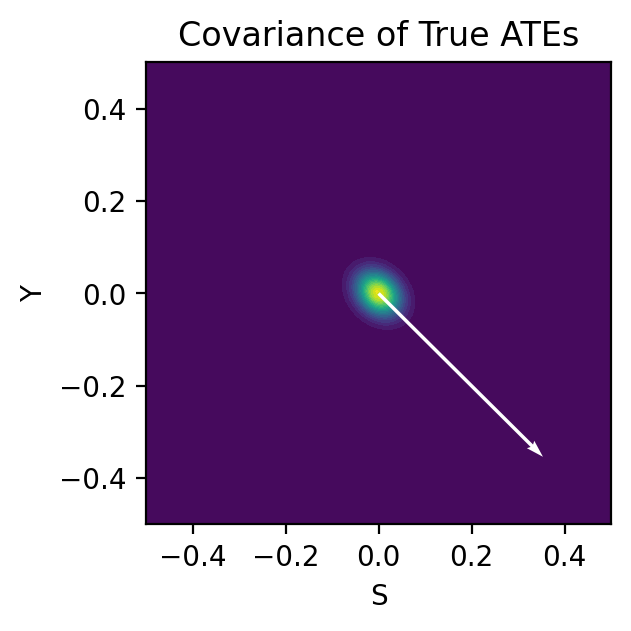}
    \includegraphics[width=2in]{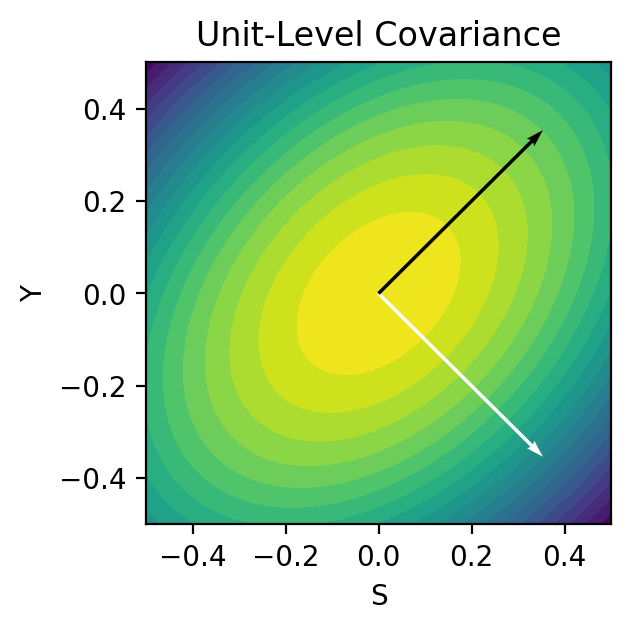}
    \includegraphics[width=2in]{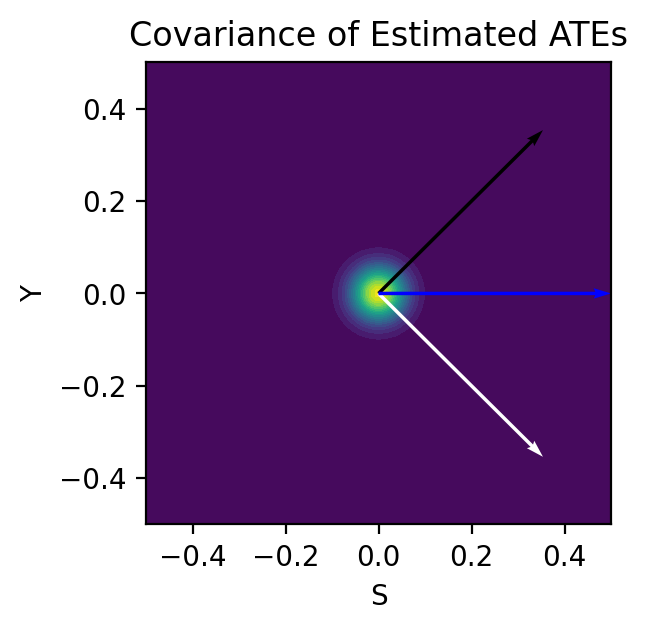}
    \caption{How Measurement Error Distorts Treatment Effect Covariances. This figure shows a hypothetical treatment effect covariance matrix, a unit-level sampling covariance matrix (``noise''), and the covariance of matrix of \emph{estimated} treatment effects, which is a weighted combination of these.}
    \label{fig:setup}
\end{figure}

As bottom panel of Figure~\ref{fig:setup} shows, the unit-level noise can overwhelm the treatment effect covariance when either the treatment effect covariance or the sample size is relatively small.  As a result, naively estimating $\theta_1(\Lambda_K)$ using the covariance matrix of the estimated treatment effects $\theta_1(\widehat \Sigma_K)$ will be biased in the direction of $\theta_1(\Omega)$.  In the absence of unit-level covariance (i.e., $\Omega=\omega I$), this bias is ``merely'' attenuation bias that preserves the direction of the relationship but biases estimates towards zero.  In the presence of unit-level covariance, the estimated covariance can be arbitrarily biased, and this bias is worse for small experiments \cite{cunningham2020interpreting}.

\subsection{LIMLK as Total Least Squares}

\begin{figure}[h!]
    \includegraphics[width=0.4\textwidth]{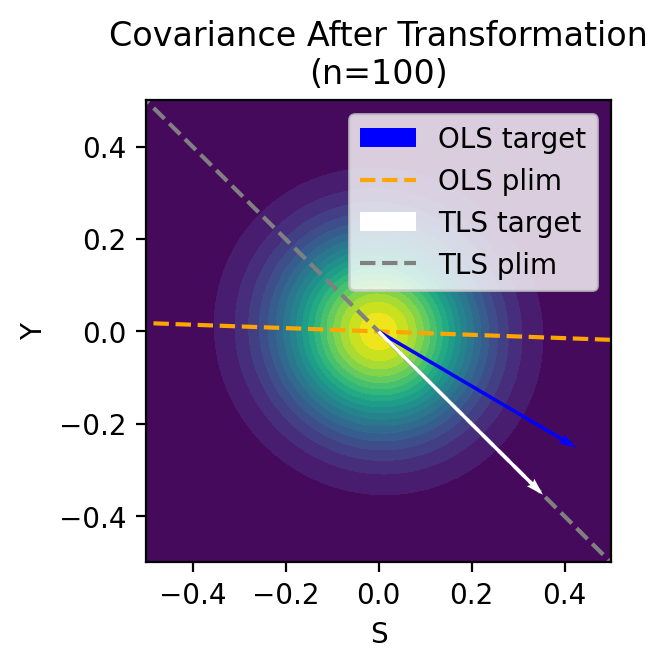}
    \includegraphics[width=0.4\textwidth]{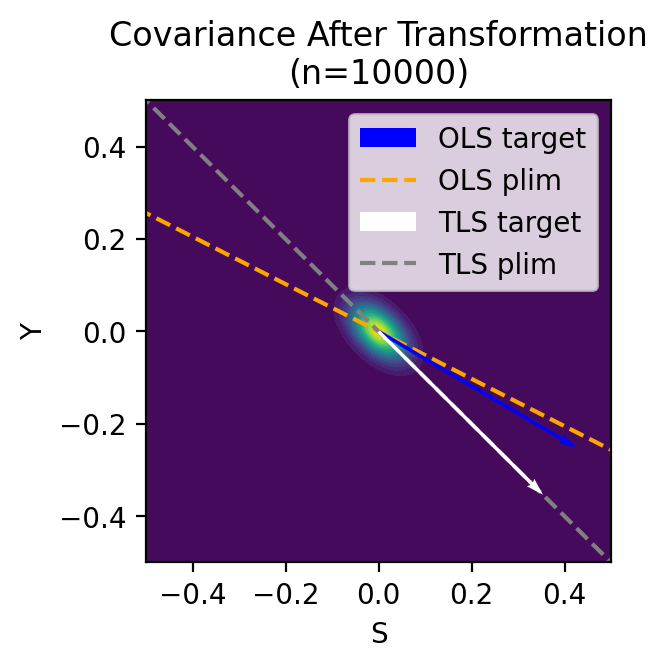}
    \caption{Effect of Transforming Estimated Treatment Effects by $\Omega^{-1/2}$. Transforming the covariance of estimated treatment effects renders the measurement error isotropic, so the TLS plim on this scatterplot coincides with the TLS target, whereas the OLS plim suffers from attenuation bias.}
    \label{fig:rotation}
\end{figure}

We now consider the LIMLK estimator as an alternative to the naive estimator based on the covariance matrix of estimated treatment effects.  As described in Section~\ref{sec:tls}, LIMLK is equivalent to performing Total Least Squares (TLS) after applying the linear transformation $\Omega^{-1/2}$ to the scatterplot of estimated treatment effects.  Intuitively, this is because applying this transformation on both sides to $\widehat \Sigma_K$ yields, in expectation, the transformed true covariance matrix $\Omega^{-1/2}\Lambda_K \Omega^{-1/2}$ plus isotropic noise --- a classic error-in-variables setup.  TLS is an effective method for addressing error-in-variables as isotropic noise does not change the eigenvectors of $\widehat \Sigma_K$.

We illustrate the effects of transforming $\widehat \Sigma_K$ by $\Omega^{-1/2}$ on least squares estimators in Figure~\ref{fig:rotation}.  In both panels, the blue arrow is the statistical parameter $\theta_1(\Omega^{-1/2} \Lambda_K \Omega^{-1/2})$, which we call the \textbf{OLS target}.  This is the OLS on the transformed scatterplot of true treatment effects.  In contrast, the orange line is the statistical parameter $\theta_1(\Omega^{-1/2} E\widehat \Sigma_K \Omega^{-1/2})$, which we call the \textbf{OLS plim}.  In words, this is the OLS on the transformed scatterplot of the estimated treatment effects as $K \to \infty$.  The white arrow and gray line in Figure~\ref{fig:rotation} are the analogous \textbf{TLS target}, $\theta_{2,\Omega}(\Lambda_K)$, and \textbf{TLS plim}, $\theta_{2,\Omega}(E\widehat \Sigma_K)$, respectively.  To illustrate the effect of sample size, the upper panel sets $n=100$, while the bottom panel sets $n=10000$.

Because transforming the scatterplot of estimated treatment effects isotropizes but does not eliminate the noise, the OLS plim of the transformed scatterplot will suffer from attenuation bias relative to the OLS target, with the magnitude of this bias decreasing in $n$ (evidenced by the convergence of the estimand to the target as $n\to\infty$).  However, isotropic noise does not change the eigenvectors of the covariance matrix, so increasing $n$ does not affect the TLS plim, which always aligns with the TLS target.  This motivates the use of TLS on the transformed scatterplot, which we show in Section $\ref{sec:tls}$ is equivalent to the LIMLK estimator.

\subsection{Comparing the Naive, LIMLK, and TC Estimators}
We now present the results of a simulation study of the empirical properties of our estimators.  The estimators evaluated are the naive OLS performed on the empirical covariance matrix of estimated effects, LIMLK, and the Total Covariance (TC) estimator, which performs OLS on the covariance matrix obtained by subtracting $(4/n) \Omega$ from $\widehat \Sigma_K$.

To illustrate the effects of unit-level noise on our estimators, we simulate two surrogates $S_1$ and $S_2$ with corresponding structural parameters $\beta_1$ and $\beta_2$.  $S_1$ is highly correlated with the primary outcome $Y$ in terms of the noise covariance matrix $\Omega$, while $S_2$ is uncorrelated with $Y$.  As a result, we expect the naive estimator to perform relatively badly for $S_1$ (as it ignores the unit-level covariance) compared to $S_2$ and the LIMLK and TC estimators.  To assess the robustness of our estimators to direct effects, we simulate data according to the linear structural models in Subsections \ref{sec:no-direct-effects} (no direct effects) and \ref{sec:direct-effects} (direct effects under INSIDE).\footnote{Code to replicate our simulations and all figures except those in Section~\ref{sec:linear_models} can be found at \url{https://github.com/winston-chou/linear-proxy-metrics}.}

As Figure \ref{fig:no-direct-effects} shows, the naive estimator is heavily biased.  LIMLK and TC are substantially much more accurate than the naive estimator.  The relative accuracy of LIMLK and TC is more apparent for $\beta_1$, which has the greater distortion from the unit-level covariance.  Without direct effects, LIMLK is more efficient compared to TC, with a considerably smaller standard error at smaller values of $n$.

However, LIMLK is highly sensitive to the assumption of no direct effects.  As Figure \ref{fig:direct-effects} shows, when we introduce direct effects, LIMLK is clearly inconsistent for $\beta$.  In contrast, TC is robust and remains consistent for $\beta$.  However, note that, regardless of estimator, the \emph{interpretation} of $\beta$ is sensitive to the presence of direct effects.  If $S$ fully mediates the effect of $A$ on $Y$, $h(S) = \beta S$ is a surrogate index, and the ATE of a new treatment on $Y$ is equal to the ATE on $h(S)$.  If $S$ does not fully mediate the effect of $A$ on $Y$, and direct effects follow the INSIDE assumption, then $\beta$ can be interpreted as the portion of the treatment effect of $Y$ that is mediated by $S$.

\begin{figure*}[p]
    \includegraphics[width=0.92\textwidth]{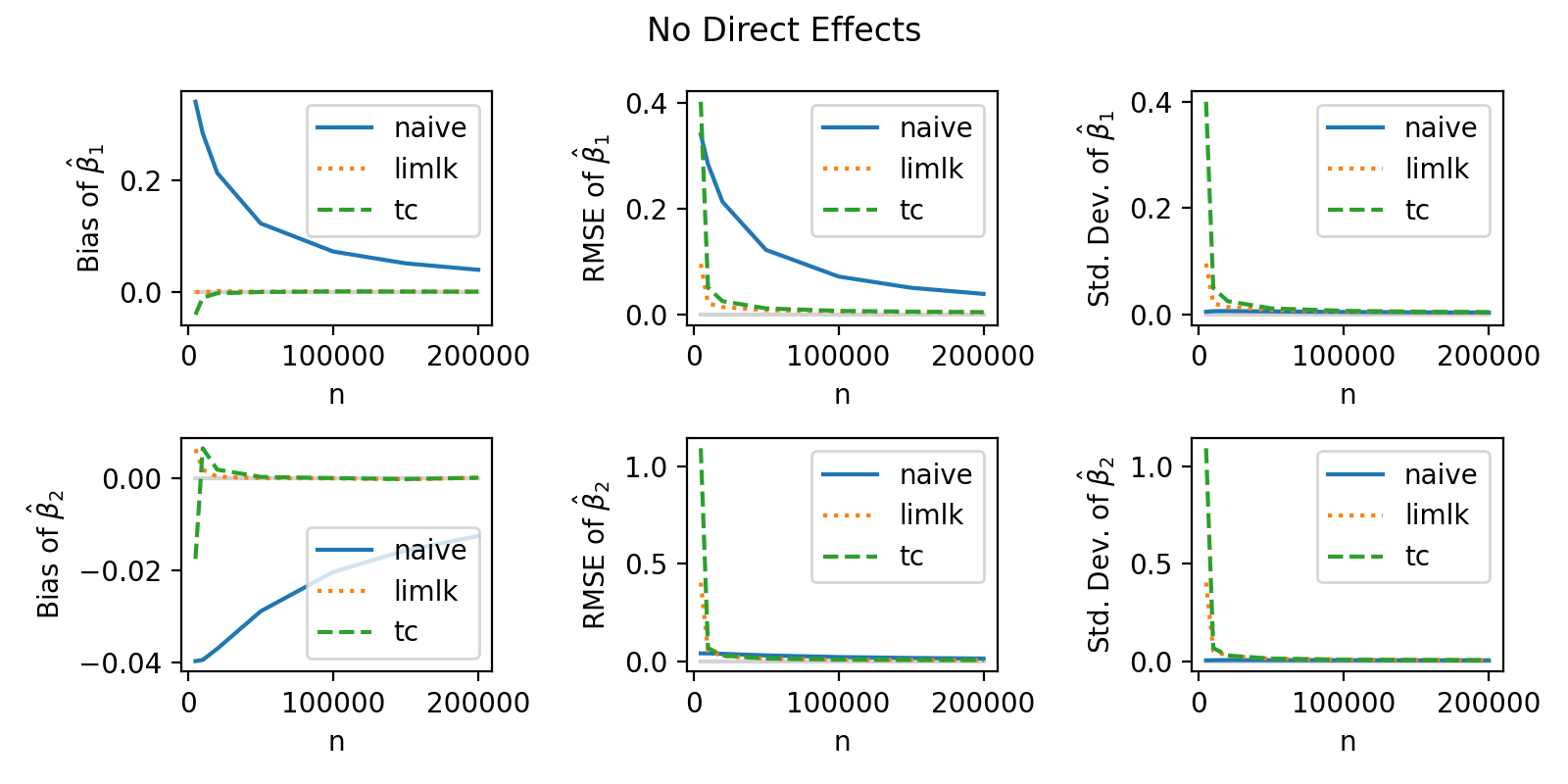}
    \caption{Comparison of Naive, LIML, and Total Covariance Estimators Without Direct Effects.  When there are no direct effects, such that the secondary outcomes fully mediate the effect of the intervention on the primary outcome, our LIMLK (\texttt{limlk}) and Total Covariance (\texttt{tc}) estimators are substantially less biased than the naive regression of the estimated ATEs on the primary outcome on the estimated ATEs on the secondary outcomes (\texttt{naive}).  LIMLK has lower variance than TC, which is most apparent at small sample sizes.}
    \label{fig:no-direct-effects}
\end{figure*}

\begin{figure*}[p]
    \includegraphics[width=0.92\textwidth]{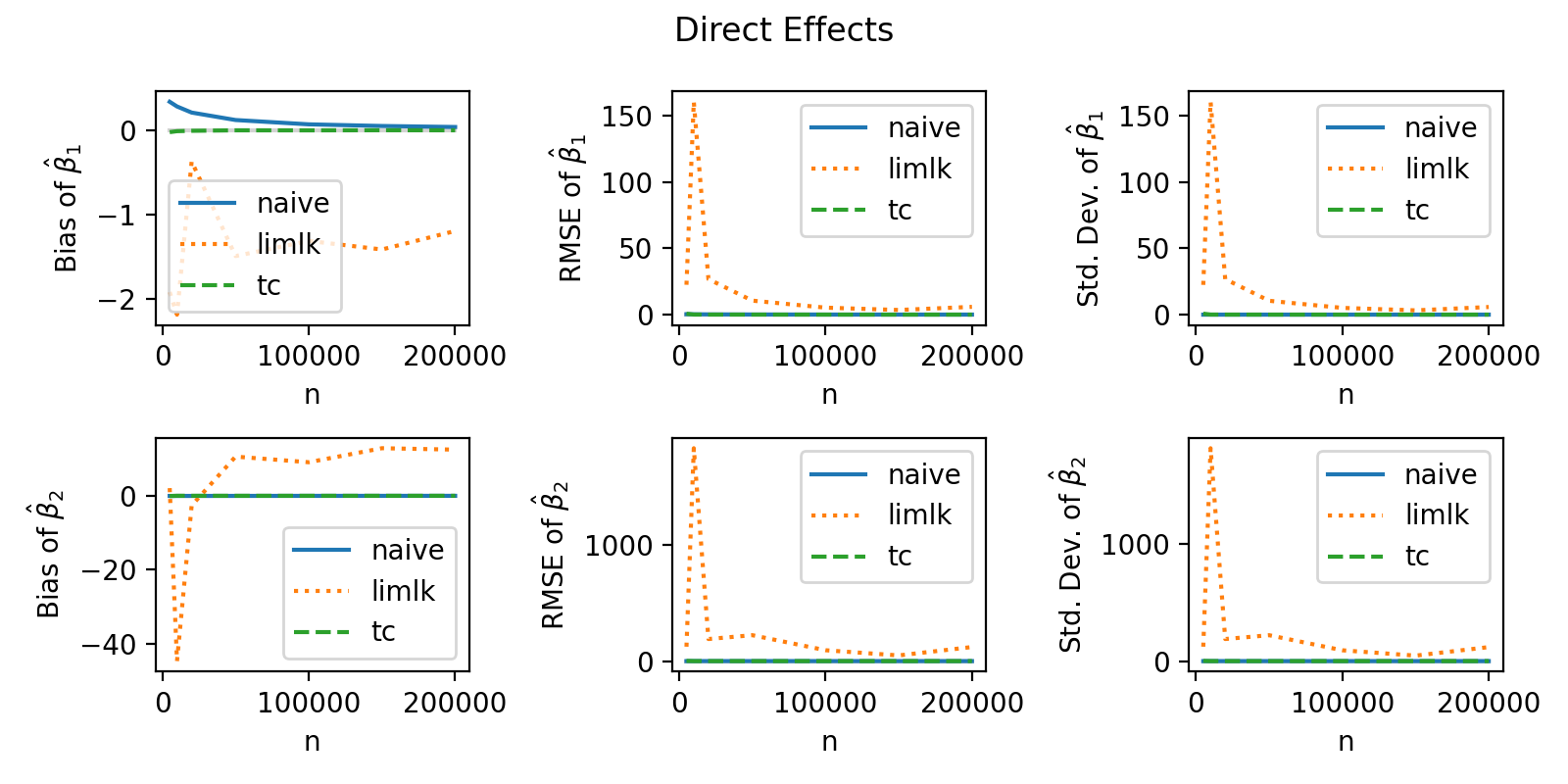}
    \caption{Comparison of Naive, LIML, and Total Covariance Estimators With Direct Effects.  When the intervention directly affects the primary outcome, LIMLK can be extremely biased.  In contrast, TC is robust to direct effects, although the causal interpretation of its estimand changes if there are direct effects.  The naive estimator is biased regardless of whether there are direct effects or not.}
    \label{fig:direct-effects}
\end{figure*}

\section{Linear Models of Treatment Effect Covariances at Netflix}\label{sec:linear_models}

\begin{figure}[!ht]
    \includegraphics[width=0.45\textwidth]{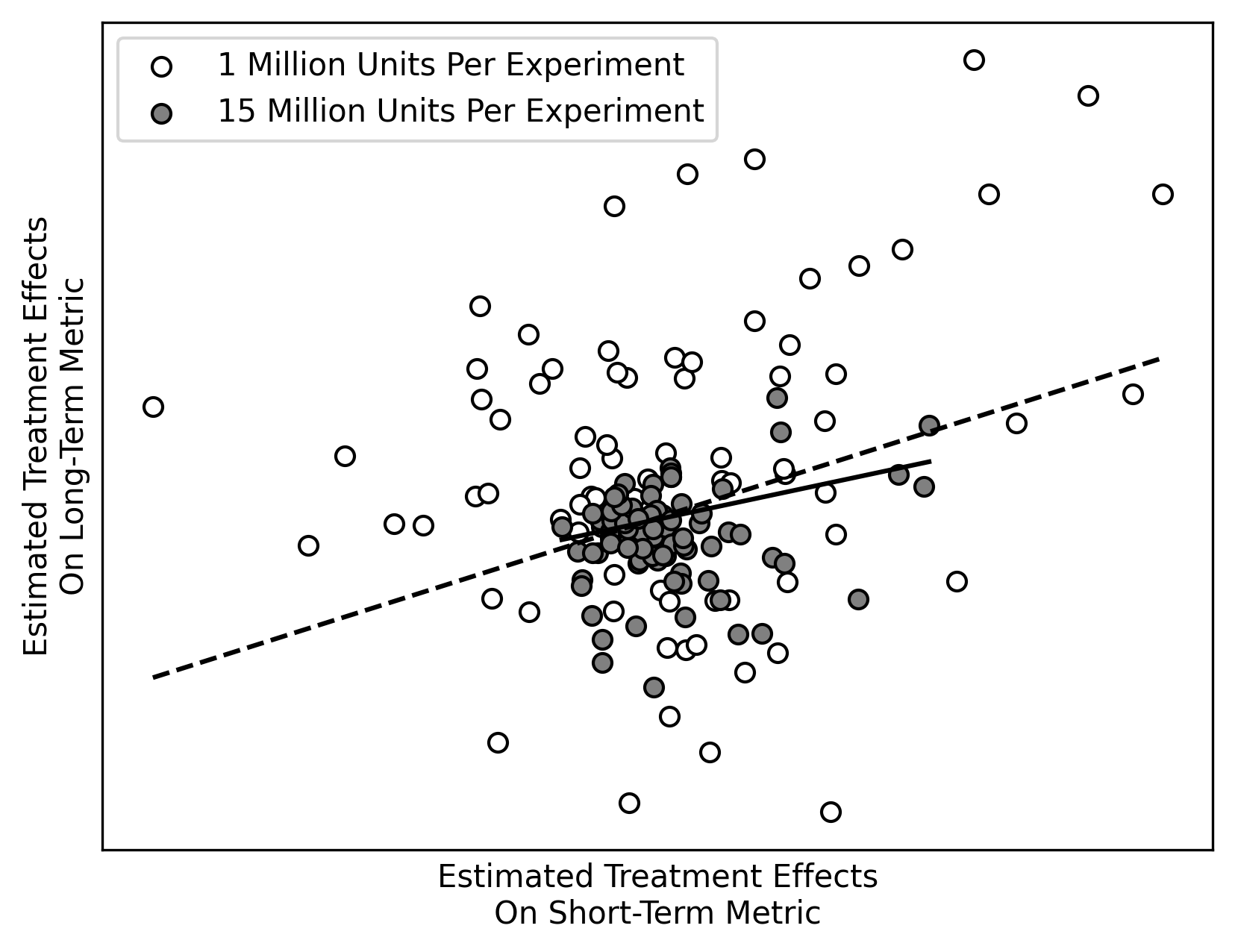}
    \includegraphics[width=0.45\textwidth]{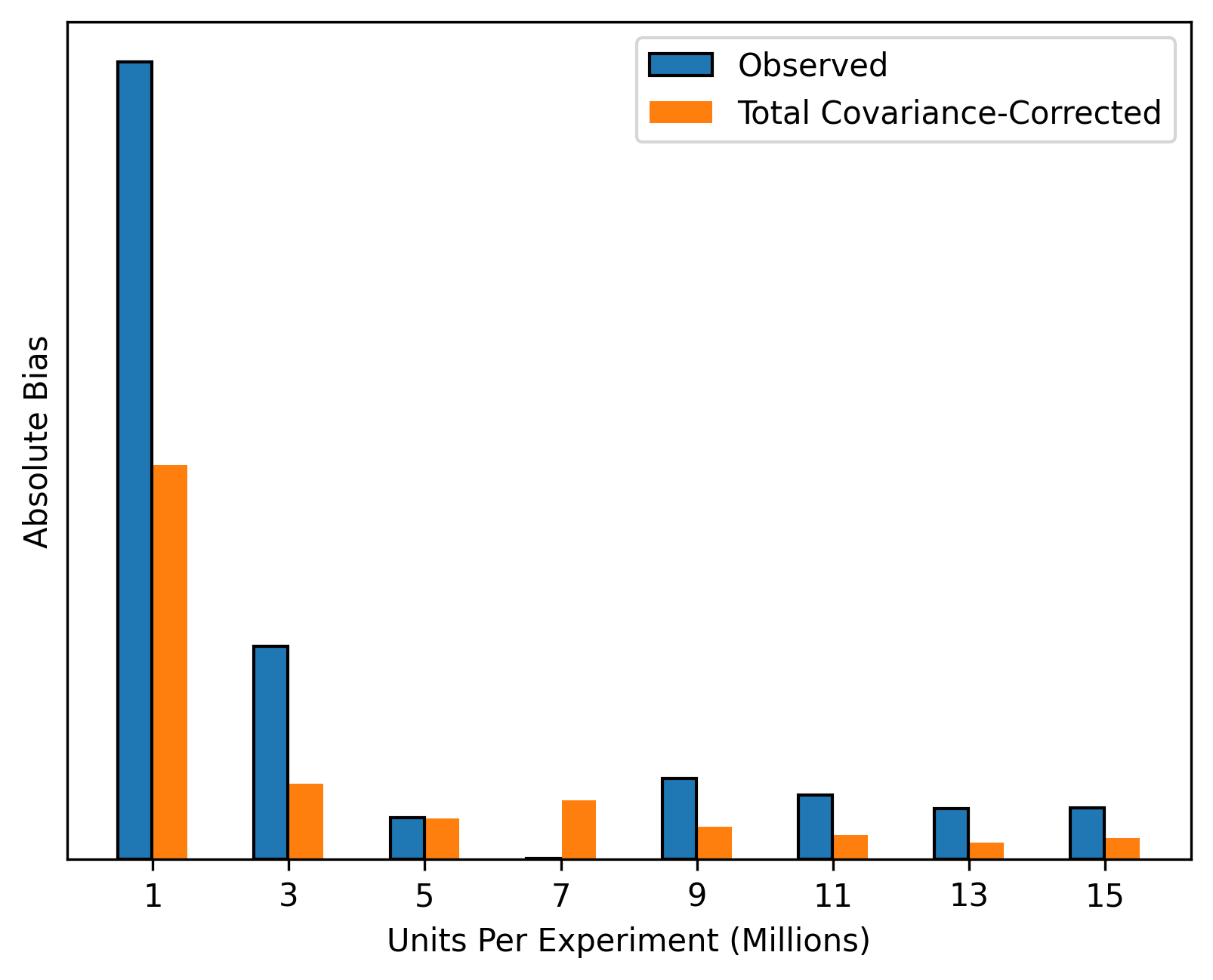}
    \caption{How Our Method Improves Treatment Effect Covariance Estimation at Netflix. This figure shows how correlated measurement error exaggerates the correlation between \emph{estimated} treatment effects, especially at smaller sample sizes (top), and how our TC estimator reduces this bias by subtracting a scaled estimate of the measurement error (bottom).}
    \label{fig:netflix}
\end{figure}

This section discusses the real-world application of our methods to construct linear proxy metric indices for experimentation at Netflix.  Netflix has a sophisticated experimentation platform that runs thousands of experiments on millions of experimental units \cite{forsell2020success}.  Still, because the signal-to-noise ratio of each experiment is small, measurement error poses challenges.  This is shown in the top panel of Figure~\ref{fig:netflix}, which plots the correlation between \emph{estimated} treatment effects on one Short-term Metric and one Long-term Metric across 96 arbitrarily-chosen treatment-control comparisons.  The Short- and Long-term Metrics are highly correlated across experimental units, introducing correlated measurement error that exaggerates the correlation across treatment effects.  To demonstrate the effect of experiment size on this bias, we show the scatterplot of estimated treatment effects after downsampling each experiment to one million units (white circles) and 15 million units (gray circles).  The slope of the OLS in the former scatterplot is about 50\% larger than the equivalent slope in the latter scatterplot.  Note that we have hidden the x- and y-axis tick marks in this figure for confidentiality reasons.

In the bottom panel of Figure~\ref{fig:netflix}, we plot the absolute bias of the observed covariance between the Short- and Long-term Metrics at various subsample sizes, using as ground truth the TC-corrected covariance across non-downsampled experiments.  We also plot the absolute bias of the TC-corrected covariance for comparison.  The median absolute bias reduction from applying our TC estimator is substantial --- approximately 63\%.

In addition to these empirical benefits, our method is also well-suited to the decentralized and rapidly-evolving practice of experimentation at Netflix.  Netflix runs thousands of experiments per year on many diverse parts of the business. Each area of experimentation is staffed by independent data science and engineering teams.  While every team ultimately aims to lift the same north star metrics (e.g., long-term revenue), each has also developed secondary metrics that are more sensitive and practical to measure in short-term A/B tests (e.g., user engagement or latency).  To complicate matters further, teams are constantly innovating on these secondary metrics to find the right balance of sensitivity and impact on the north star metrics.

In this decentralized environment, linear models of relationships between treatment effects are a highly useful tool for coordinating these efforts and aligning them towards the same proxies for the north star:
\begin{enumerate}
    \item \textbf{Managing metric tradeoffs.} Because experiments in one area can affect metrics in another area, there is a need to measure all secondary metrics in all tests, but also to understand the relative impact of these metrics on the north star.  This is so we can inform decision-making when one metric trades off against another metric. 
    \item \textbf{Informing metrics innovation.}  To minimize wasted effort on metric development, it is also important to understand how metrics correlate with the north star ``net of'' existing metrics.  
    \item \textbf{Enabling teams to work independently.}  Lastly, teams need simple tools in order to iterate on their own metrics.  Teams may come up with dozens of variations of secondary metrics, and slow, complicated tools for evaluating these variations are unlikely to be adopted.
\end{enumerate}

Given these needs and the availability of statistics from historical experiments, it has been common for teams to fit linear models to estimated treatment effects to identify promising surrogates and estimate their weights in a linear proxy metric index.  While data scientists are aware that these models can be highly biased, linear models have a convenience and interpretability that is difficult to replace.  Our TC estimator provides a simple way to consistently and robustly estimate the true treatment effect covariance matrix, which supports the construction of these linear proxies, and is actively used to develop proxy metrics at Netflix.

\section{Conclusion}

In this paper, we discuss the construction of proxy metrics via meta-analysis of many experiments.  A useful parameter in this setting is the covariance matrix of treatment effects across metrics in the population of experiments.  While estimating this parameter is computationally convenient, it presents challenges, especially when treatment effects are small relative to unit-level noise.  We present estimators for overcoming these challenges and show how their estimands relate to structural parameters under different causal models.

Throughout, we have assumed that the unit-level covariance matrix is correctly specified and constant across experiments.  This assumption is often reasonable given that it can be estimated on vastly more units than are in any individual experiment (e.g., across the entire user base).  In cases where it is not reasonable, users can also compute Jackknife covariance matrices.  That being said, this is a computationally expensive approach, and future work can explore more scalable ways of relaxing this assumption.

We show that LIMLK is efficient but inconsistent under direct effects while TC is consistent under direct and indirect effects.  This suggests that it is possible to construct a statistical test for direct effects by comparing the LIMLK and TC residuals.  We also leave this as a direction for future research.

Beyond the construction of proxy metrics, there are other opportunities to leverage short-term observations for inference on long-term outcomes. For example, \cite{athey2020combining,imbens2022long,chen2023semiparametric} consider how short-term observations from one randomized experiment can remove confounding in an observational study with long-term observations, and \cite{kallus2020role} consider efficiency gains from including units with no long-term observations. Considering the use of many, albeit weak, experiments as instruments in these settings may uncover new opportunities to target long-term outcomes under weaker assumptions.

\balance

\bibliography{bib}

\begin{thebibliography}{10}

\bibitem{anderson2009limited}
{\sc Anderson, T., Kunitomo, N., Matsushita, Y., et~al.}
\newblock The limited information maximum likelihood estimator as an angle.
\newblock {\em CIRJE No. CIRJE-F-619, CIRJE, Faculty of Economics, University of Tokyo\/} (2009).

\bibitem{angrist1999jackknife}
{\sc Angrist, J.~D., Imbens, G.~W., and Krueger, A.~B.}
\newblock Jackknife instrumental variables estimation.
\newblock {\em Journal of Applied Econometrics 14}, 1 (1999), 57--67.

\bibitem{athey2020combining}
{\sc Athey, S., Chetty, R., and Imbens, G.}
\newblock Combining experimental and observational data to estimate treatment effects on long term outcomes.
\newblock {\em arXiv preprint arXiv:2006.09676\/} (2020).

\bibitem{athey2016estimating}
{\sc Athey, S., Chetty, R., Imbens, G., and Kang, H.}
\newblock Estimating treatment effects using multiple surrogates: The role of the surrogate score and the surrogate index.
\newblock {\em arXiv preprint arXiv:1603.09326\/} (2016).

\bibitem{burgess2017interpreting}
{\sc Burgess, S., and Thompson, S.~G.}
\newblock Interpreting findings from mendelian randomization using the mr-egger method.
\newblock {\em European journal of epidemiology 32\/} (2017), 377--389.

\bibitem{chen2023semiparametric}
{\sc Chen, J., and Ritzwoller, D.~M.}
\newblock Semiparametric estimation of long-term treatment effects.
\newblock {\em Journal of Econometrics 237}, 2 (2023), 105545.

\bibitem{cunningham2020interpreting}
{\sc Cunningham, T., and Kim, J.}
\newblock Interpreting experiments with multiple outcomes.

\bibitem{deng2013improving}
{\sc Deng, A., Xu, Y., Kohavi, R., and Walker, T.}
\newblock Improving the sensitivity of online controlled experiments by utilizing pre-experiment data.
\newblock In {\em Proceedings of the sixth ACM international conference on Web search and data mining\/} (2013), pp.~123--132.

\bibitem{elliott2015surrogacy}
{\sc Elliott, M.~R., Conlon, A.~S., Li, Y., Kaciroti, N., and Taylor, J.~M.}
\newblock Surrogacy marker paradox measures in meta-analytic settings.
\newblock {\em Biostatistics 16}, 2 (2015), 400--412.

\bibitem{forsell2020success}
{\sc Forsell, E., Beckley, J., Ejdemyr, S., Hannan, V., Rhines, A., Tingley, M., Wardrop, M., and Wong, J.}
\newblock Success stories from a democratized experimentation platform.
\newblock {\em arXiv preprint arXiv:2012.10403\/} (2020).

\bibitem{hahn2004estimation}
{\sc Hahn, J., Hausman, J., and Kuersteiner, G.}
\newblock Estimation with weak instruments: Accuracy of higher-order bias and mse approximations.
\newblock {\em The Econometrics Journal 7}, 1 (2004), 272--306.

\bibitem{hohnhold2015focusing}
{\sc Hohnhold, H., O'Brien, D., and Tang, D.}
\newblock Focusing on the long-term: It's good for users and business.
\newblock In {\em Proceedings of the 21th ACM SIGKDD International Conference on Knowledge Discovery and Data Mining\/} (2015), pp.~1849--1858.

\bibitem{imbens2022long}
{\sc Imbens, G., Kallus, N., Mao, X., and Wang, Y.}
\newblock Long-term causal inference under persistent confounding via data combination.
\newblock {\em arXiv preprint arXiv:2202.07234\/} (2022).

\bibitem{kallus2020role}
{\sc Kallus, N., and Mao, X.}
\newblock On the role of surrogates in the efficient estimation of treatment effects with limited outcome data.
\newblock {\em arXiv preprint arXiv:2003.12408\/} (2020).

\bibitem{mikusheva2022inference}
{\sc Mikusheva, A., and Sun, L.}
\newblock Inference with many weak instruments.
\newblock {\em The Review of Economic Studies 89}, 5 (2022), 2663--2686.

\bibitem{peysakhovich2018learning}
{\sc Peysakhovich, A., and Eckles, D.}
\newblock Learning causal effects from many randomized experiments using regularized instrumental variables.
\newblock In {\em Proceedings of the 2018 World Wide Web Conference\/} (2018), pp.~699--707.

\bibitem{prentice1989surrogate}
{\sc Prentice, R.~L.}
\newblock Surrogate endpoints in clinical trials: definition and operational criteria.
\newblock {\em Statistics in medicine 8}, 4 (1989), 431--440.

\bibitem{tripuraneni2023choosing}
{\sc Tripuraneni, N., Richardson, L., D'Amour, A., Soriano, J., and Yadlowsky, S.}
\newblock Choosing a proxy metric from past experiments.
\newblock {\em arXiv preprint arXiv:2309.07893\/} (2023).

\bibitem{wong2019efficient}
{\sc Wong, J., Lewis, R., and Wardrop, M.}
\newblock Efficient computation of linear model treatment effects in an experimentation platform.
\newblock {\em arXiv preprint arXiv:1910.01305\/} (2019).

\bibitem{xie2016improving}
{\sc Xie, H., and Aurisset, J.}
\newblock Improving the sensitivity of online controlled experiments: Case studies at netflix.
\newblock In {\em Proceedings of the 22nd ACM SIGKDD International Conference on Knowledge Discovery and Data Mining\/} (2016), pp.~645--654.

\end{thebibliography}






\end{document}